\documentclass{amsart}
\usepackage{amssymb,stmaryrd,mathrsfs,dsfont}

\theoremstyle{plain}
\newtheorem{theorem}{Theorem}[section]

\newtheorem{proposition}[theorem]{Proposition}
\newtheorem{corollary}[theorem]{Corollary}

\theoremstyle{definition}
\newtheorem{notation}[theorem]{Notation}

\newtheorem{definition}[theorem]{Definition}

\theoremstyle{remark}
\newtheorem{remark}[theorem]{Remark}

\input xy
\xyoption{all}

\begin{document}
\title[$L$--Primitive Words in Submonoids]{$L$--Primitive Words in Submonoids}
\author[S. N. Singh]{Shubh Narayan Singh}
\address{Department of Mathematics, Central University of Bihar, Patna, India}
\email{shubh@cub.ac.in}

\author[K. V. Krishna]{K. V. Krishna}
\address{Department of Mathematics, IIT Guwahati, Guwahati, India}
\email{kvk@iitg.ac.in}


\begin{abstract}
This work considers a natural generalization of primitivity with respect to a language. Given a language $L$, a nonempty word $w$ is said to be $L$-primitive if $w$ is not a proper power of any word in $L$. After ascertaining the number of primitive words in submonoids  of a free monoid,  the work  proceeds to count $L$-primitive words in submonoids of a free monoid. The work also studies the distribution of $L$-primitive words in certain subsets of free monoids.
\end{abstract}



\maketitle

\section*{Introduction}

A nonempty word which is not a power of any other word is called a primitive word. It is well known that every nonempty word can be uniquely expressed as a power of a primitive word \cite{lyndon62}. The study of primitivity of words is often the first step towards the understanding of words and plays an important role in the theory of languages. Ito \emph{et al.} have investigated the number of primitive words in the languages accepted by automata \cite{ito88}.  Shyr and Tseng have proved that any noncommutative submonoid of a free monoid contains infinitely many primitive words \cite{shyr84}. In the literature, there are various types of generalizations/extensions of the classic definition of primitive words \cite{czeizler10,michael04,hsiao02,kari98}. We propose yet another generalization of primitive word, viz. $L$-primitive word -- a  nonempty word that is not a proper power of any word in a given language $L$.

In this paper, we first investigate the primitive words in the submonoids of free monoids. We could ascertain that the number of primitive words in a submonoid  of a free monoid is either at most one or infinity. Then, we study the distribution of $L$-primitive words in certain subsets of free monoids. In particular, we target to count the $L$-primitive words in the submonoids of free monoids.

The paper is organized as follows. In Section 1, we will present some necessary preliminaries of the paper. We introduce the concept of $L$-primitive words in Section 2 and study some basic properties. In Section 3, we count the number of $L$-primitive words in a language $L$ as well as in the submonoids of a free monoid. Finally, Section 4 concludes the paper.

\section{Preliminaries}

In this section, we present some basic definitions and fix our notations. For more details one may refer \cite{berstel85,lothaire83,shubh12}.

Let $A$ be a nonempty finite set called an \emph{alphabet} with its elements as \emph{letters}. The free monoid over $A$ is denoted by $A^*$ whose elements are called words, and $\varepsilon$ denotes the identity element of $A^*$ -- the empty word.  The set of all nonempty words over $A$ is denoted by $A^+$, i.e. $A^+ = A^*\setminus \{\varepsilon\}$ .

A word $u$ is said to be a \emph{prefix} of a word $v$ if there exists a word $t$ such that $ut = v$. A set $X$ of words is called a \emph{prefix set} if no element of $X$ is a prefix of another word of $X$. A \emph{power} of a word $u$ is a word of the form $u^k$ for some $k \in \mathds{N} = \{0, 1, 2, \ldots\}$ -- the set of natural numbers. It is convenient to set $u^0 = \varepsilon$, for each word $u$. If $k \in \mathds{N} \setminus \{0, 1\}$, we say that $u^k$ is a \emph{proper power} of $u$. A word $x \in A^+$ is said to be a \emph{primitive word} if it is not a proper power of another word in $A^*$, i.e. for $u \in A^*$, \[x = u^k \;\Longrightarrow\; k = 1.\]
For a subset $X$ of $A^*$, we denote the set of all primitive words in $X$ by $X_p$.  We recall the following well known property of primitive words.

\begin{proposition}\label{c5.p.pr}
For every $w \in A^+$,  there exists a unique primitive word $u$ and a unique integer $k \ge 1$ such that $w = u^k$.
\end{proposition}

The unique primitive word $u$ obtained in Proposition \ref{c5.p.pr} is called the \emph{primitive root} of $w$, denoted by $\sqrt{w}$. By a language over an alphabet $A$ is meant a subset of $A^*$. The \emph{root} of a language $L$ ($\subseteq A^*$), denoted by $\sqrt{L}$, is defined as \[\sqrt{L} = \big\{\sqrt{w} \in A_p^*\;\big|\; w \in L\setminus\{\varepsilon\}\big\}.\] A language $L$ is said to be \emph{commutative} if $uv = vu$, for all $u,v \in L$. It is known that a language $L$ is commutative if and only if there exists $w \in A^*$ such that $L \subseteq \{w\}^*$.

Now, we recall some properties of numerical monoids from \cite{rosales09}. A \emph{numerical monoid} is a submonoid of the monoid  $(\mathds{N}, +)$ whose complement in $\mathds{N}$ is finite. For a nonempty subset $X$ of $\mathds{N}$, the submonoid of $\mathds{N}$ generated by $X$ is denoted by $\langle X \rangle$, i.e.
\[\langle X \rangle = \{\lambda_1 x_1 + \cdots + \lambda_n x_n \;|\; n, \lambda_i \in \mathds{N}, x_i \in X, \forall i( 1 \le i \le n)\}.\]

\begin{theorem}\label{c5.t.numimonoidgcd1}
Let $X$ be a nonempty subset of $\mathds{N}$. The submonoid $\langle X \rangle$ is a numerical monoid if and only if $\gcd(X) = 1$.
\end{theorem}

\begin{theorem}\label{c5.t.num-umgsf}
Every numerical monoid admits a unique finite minimal set of generators.
\end{theorem}

\begin{theorem}\label{c5.t.smisonm}
Any nontrivial submonoid of $\mathds{N}$ is isomorphic to a numerical monoid.
\end{theorem}

\section{$L$-primitive words}

In this section, we introduce the notion of primitive words relative to a language $L$, called $L$-primitive words, and obtain some properties related to $L$-primitive words. We prove that every primitive word is an $L$-primitive word so that the latter notion is a generalization of the former one. Unless it is specified otherwise, in what follows, $L$ is an arbitrary language over $A$.

\begin{definition}
A word $x \in A^+$ is said to be an $L$-\emph{primitive word} if $x$ is not a proper power of any word in $L$, i.e., for $u \in L$, \[x = u^k \;\Longrightarrow\; k = 1.\]
\end{definition}

\begin{notation}  Let $X \subseteq A^*$ and $X^c$ denotes the complement of $X$ in $A^*$.
\begin{itemize}
\item[(i)] The set of $L$-primitive words in $X$ is denoted by $X_{L\mbox{-}p}$.
\item[(ii)] The set $(X^*)_{L\mbox{-}p}$ of $L$-primitive words in $X^*$ is simply denoted by $X^*_{L\mbox{-}p}$.
\item[(iii)] The set $(X^c)_{L\mbox{-}p}$ of $L$-primitive words in $X^c$ is simply denoted by $X^c_{L\mbox{-}p}$.
\end{itemize}
\end{notation}

We begin with some basic properties of $L$-primitive words.

\begin{remark}\label{c5.r.ALpQ}\
\begin{enumerate}
\item[\rm(i)] If $L = \varnothing$, then $A^*_{L\mbox{-}p} = A^+$, the set of all nonempty words over $A$.
\item[\rm(ii)] If $L = A^*$, then $A^*_{L\mbox{-}p} = A_p^*$, the set of all primitive words over $A$.
\end{enumerate}
\end{remark}

\begin{proposition}\label{c5.p.comp-lpri}
If $L_1$ and $L_2$ are two subsets of $A^*$, then \[L_1 \subseteq L_2 \;\Longrightarrow \; A^*_{L_2\mbox{-}p} \subseteq A^*_{L_1\mbox{-}p}.\]
\end{proposition}

\begin{proof}
On the contrary, let us assume that $A^*_{L_2\mbox{-}p} \not\subseteq A^*_{L_1\mbox{-}p}$. Then there exists \break $w \in A^*_{L_2\mbox{-}p}$, but $w \not\in A^*_{L_1\mbox{-}p}$. Since $w \not\in A^*_{L_1\mbox{-}p}$, there exists $u \in L_1$ such that $w = u^k$, for some $k > 1$. In view of hypothesis, we have $u \in L_2$. Consequently, $w \notin A^*_{L_2\mbox{-}p}$; a contradiction.
\end{proof}

In view of Remark \ref{c5.r.ALpQ}(ii), we have the following corollary of Proposition \ref{c5.p.comp-lpri}.

\begin{corollary}\label{c5.c.prim-lprim}
Every primitive word is an $L$-primitive word. Hence, if $|A|\ge 2$, then $|A^*_{L\mbox{-}p}| = \infty$.
\end{corollary}

\begin{remark}
An $L$-primitive word need not be primitive. For instance, let $L = \{abab\} \subseteq \{a, b\}^*$. Clearly, the word $abab$ is an $L$-primitive word, but not a primitive word.
\end{remark}

\begin{definition}
For $w \in A^+$, we define the set of $L$-\emph{primitive roots} of $w$, denoted by $\sqrt[L]{w}$, is defined as
\[\sqrt[L]{w} = \{x \in A^*_{L\mbox{-}p}\;|\; x^k = w, \;\mbox{for some}\; k\geq 1\}.\]
Further, for $X  \subseteq A^*$, the $L$-\emph{primitive root} of $X$, denoted by $\sqrt[L]{X}$, is defined as \[\sqrt[L]{X} = \bigcup_{w \in X\setminus \{\varepsilon\}}\sqrt[L]{w}.\]
\end{definition}

\begin{remark}
The primitive root of a nonempty word is an $L$-primitive root of the word. Thus, if $w \ne \varepsilon$, then $\sqrt[L]{w} \ne \varnothing$.
\end{remark}

\section{$L$-primitive words in various subsets}

This section is divided into two subsections. We investigate $L$-primitive words in some sets related to $L$ itself in Subsection 3.1. Then, we carry on the investigations  on submonoids in Subsection 3.2.

\subsection{$L$-primitive words in $L$}

In this subsection, we make an attempt to investigate $L$-primitive words in  $L$ and also in $L^c$. In this connection, we provide some sufficient conditions and characterizations. In fact, we give a relation between $L$-primitive words and $L$-primitive roots in $L$.

\begin{theorem}
If $\varepsilon \not\in L$, $L \neq \varnothing$ if and only if $L_{L\mbox{-}p} \neq \varnothing$.
\end{theorem}

\begin{proof}
Let us assume that $L \neq \varnothing$ and choose $w\in L$.  If $w \in L_{L\mbox{-}p}$, then we are through. Otherwise, there exists $u \in L$ such that $w = u^k$, for some $k > 1$.  Clearly, $|u| < |w|$. If $u \in L_{L\mbox{-}p}$, then we are through. Otherwise, we continue to choose shorter words in $L$ whose power is $w$. But this process terminates at a finite stage and eventually we get a word $x \in L_{L\mbox{-}p}$ and $w = x^m$, for some $m > 1$. Hence, $L_{L\mbox{-}p} \neq \varnothing$. The converse is straightforward.
\end{proof}

\begin{theorem}\label{c5.t.prefix}
If $L \subseteq A^+$ is a prefix set, then $L = L_{L\mbox{-}p}$.
\end{theorem}

\begin{proof}
Clearly, $ L_{L\mbox{-}p} \subseteq L$. Let $x \in L$, but $x \notin  L_{L\mbox{-}p}$. There exists a word $u \in L$ such that $x = u^k$, for some $k > 1$. Thus, the word $u \in L$ is a prefix of the word $x \in L$. This contradicts that $L$ is a prefix set. Hence, $L = L_{L\mbox{-}p}$.
\end{proof}

\begin{remark}
The converse of Theorem \ref{c5.t.prefix} is not necessarily true. For instance, let $L = \{a, ab\} \subseteq \{a, b\}^+$. Clearly, $L = L_{L\mbox{-}p}$, but $L$ is not a prefix set.
\end{remark}

It is clear that $L_{L\mbox{-}p} \subseteq \sqrt[L]{L}$. Now, we explore the possibilities so that $L_{L\mbox{-}p} = \sqrt[L]{L}$. For this, we need the notion of power of a subset of $A^*$ introduced by Calbrix and Nivat (cf. \cite{calbrix96}).
The \emph{power} of a subset $X$ of $A^*$, denoted by $\mbox{\rm pow}(X)$, is defined as \[\mbox{\rm pow}(X) = \{x^k \;|\; x \in X \;\mbox{and}\; k \geq 1\}.\]

\begin{remark}
Clearly, $\mbox{\rm pow}(A^*_{L\mbox{-}p}) = A^+$.
\end{remark}

\begin{theorem} \label{c5.t.llp-ll}\
\begin{enumerate}
\item[\rm(i)] $ \sqrt[L]{L} \subseteq L \;\Longleftrightarrow\; L_{L\mbox{-}p} = \sqrt[L]{L}$.
\item[\rm(ii)] $L^c = \mbox{\rm pow}(L^c) \; \Longrightarrow \; L_{L\mbox{-}p} = \sqrt[L]{L}$.
\item[\rm(iii)] $L \subseteq A_p^*\;\Longrightarrow\; L_{L\mbox{-}p} = \sqrt[L]{L} = L$.
\end{enumerate}
\end{theorem}

\begin{proof} We first note that $L_{L\mbox{-}p} \subseteq \sqrt[L]{L}$.
\begin{itemize}
\item[(i)] (:$\Longleftarrow$) Since $L_{L\mbox{-}p} \subseteq L$, from the hypothesis, we have $\sqrt[L]{L} \subseteq L$.\\
($\Longrightarrow$:) Let $x \in \sqrt[L]{L}$; then $x$ is $L$-primitive word. Also, from the hypothesis, we have $x \in L$. Thus, $x \in L_{L\mbox{-}p}$. Hence, we have the part (i).

\item[(ii)] Let us assume  that $x\in \sqrt[L]{L}\setminus L_{L\mbox{-}p}$. Since $x$ is an $L$-primitive word and $x \notin L_{L\mbox{-}p}$, we have $x \notin L$. Then, from the hypothesis, we have $x^k \in L^c$, for all $k \ge 1$. But, since $x\in \sqrt[L]{L}$, we have $x \in \sqrt[L]{w}$, for some $w \in L$. That is, there is a number $t \ge 1$, such that $x^t = w (\in L)$; a contradiction. Hence, $\sqrt[L]{L} = L_{L\mbox{-}p}$.

\item[(iii)] Clearly, $L_{L\mbox{-}p} \subseteq L$. Let $x \in L$; from the hypothesis, we have $x \in A_p^*$. By Corollary \ref{c5.c.prim-lprim}, since every primitive word is an $L$-primitive word, we have $x \in L_{L\mbox{-}p}$. Thus, $L = L_{L\mbox{-}p}$.

It is clear that for $w \in A_p^*$, we have $\sqrt[L]{w} = \{w\}$. Since $L \subseteq A_p^*$, we have
 \[\sqrt[L]{L} = \bigcup_{w \in L}\sqrt[L]{w} = \bigcup_{w \in L}\{w\} = L.\]
Hence, if $L \subseteq A_p^*$, we have $ L_{L\mbox{-}p} = \sqrt[L]{L} = L$.
\end{itemize}
\end{proof}

\begin{corollary}
$L = \sqrt[L]{L} \: \Longleftrightarrow \;  L_{L\mbox{-}p} = \sqrt[L]{L} = L$.
\end{corollary}

\begin{remark}
The converse of Theorem \ref{c5.t.llp-ll}(ii) is not necessarily true. For instance, consider $L = \{a, b, a^6\} \subseteq \{a, b\}^+$. Observe that $L_{L\mbox{-}p} = \sqrt[L]{L} = \{a, b\}$. Clearly, since $a^2 \in L^c$, we have $a^6 \in \mbox{\rm pow}(L^c)$; but, $a^6 \notin L^c$. Hence,  $L^c \ne \mbox{\rm pow}(L^c)$.
\end{remark}

\begin{theorem}
$L = \mbox{\rm pow}(L) \; \Longleftrightarrow \; L_{L\mbox{-}p}^c = L^c$.
\end{theorem}

\begin{proof}\
\begin{itemize}
  \item[($\Longrightarrow$:)] Clearly, $L_{L\mbox{-}p}^c \subseteq L^c$. Let $x \in L^c$, but $x \notin L_{L\mbox{-}p}^c$. There exists a word $y \in L$ such that $x = y^k$, for some $k > 1$. Since $y \in L$, we have $y^k \in \mbox{\rm pow}(L)$. It follows that $x \in \mbox{\rm pow}(L)$. But, $L = \mbox{\rm pow}(L)$, we have $x \in L$. This is a contradiction.
  \item[(:$\Longleftarrow$)]  Clearly, $L \subseteq \mbox{\rm pow}(L)$. Let $x \in \mbox{\rm pow}(L)$, but $x \notin L$. There exists a word $y \in L$ such that $x = y^k$, for some $k > 1$. Since, $x \notin L$, we have $x \in L^c$. But, $L_{L\mbox{-}p}^c = L^c$, it follows that $x$ is an $L$-primitive word; which is a contradiction.
\end{itemize}
\end{proof}

\subsection{$L$-primitive words in submonoids}

In this subsection, for comparison, we first investigate the number of primitive words in the submonoids of a free monoid. Further, we study the $L$-primitive words in the submonoids of a free monoid. We count the $L$-primitive words in a submonoid of the free monoid over a unary alphabet. In this case, when $L$ is finite, we prove that a submonoid has either at most one or infinitely many $L$-primitive words. Finally, we leave certain remarks on estimating the number of $L$-primitive words over an arbitrary alphabet. We require the following theorem.

\begin{theorem}[\cite{shyr84}]\label{c5.t.prim-prop}
Let $H$ be a submonoid of $A^*$. $H$ is noncommutative if and only if $|H_p| = \infty$.
\end{theorem}

Now, we observe that a submonoid of $A^*$ contains either at most one primitive word or infinitely many primitive words.

\begin{theorem}
Let $H$ be a submonoid of $A^*$; then either $|H_p|\leq 1$ or $|H_p| = \infty$.
\end{theorem}

\begin{proof}
If $H = \{\varepsilon\}$, then $|H_p| = 0$. Let us assume that $H \ne \{\varepsilon\}$. If $H$ is noncommutative, then by Theorem \ref{c5.t.prim-prop}, we have $|H_p| = \infty$. Otherwise, we have $H \subseteq \{w\}^*$, for some word $w \in A^+$. Without loss of generality, assume that $w \in A_p^*$. Thus, according to $w \in H$ or not, we have $|H_p| = 1 \;\mbox{or}\; 0$.
\end{proof}

\begin{corollary}
If $H$ is a nontrivial submonoid of $A^*$, then either $|\sqrt{H}| = 1 \;\mbox{or}\; \infty$.
\end{corollary}

Let $A = \{a\}$ be a unary alphabet. It is known that $A^*$ is isomorphic to the additive monoid of natural numbers $(\mathds{N}, +)$ under the isomorphism given by $a^k \mapsto k$. Thus, each word $a^k$ of $A^*$ is characterized by its length $k \in \mathds{N}$.
Hence,  we count the $L$-primitive words in the submonoids of $\mathds{N}$, instead of $A^*$. In what follows, $H$ is a nontrivial submonoid of $\mathds{N}$ and $L$ is a nonempty subset of $\mathds{N}$. Now, we count the number of $L$-primitive words in $H$. We begin with the following remark.

\begin{remark}
If $1 \in L$, then according to $1 \in H$ or not, we have $|H_{L\mbox{-}p}| = 1$ or $|H_{L\mbox{-}p}| = 0$, respectively.
\end{remark}

Let us assume that $1 \notin L$. In view of Theorem \ref{c5.t.smisonm} and Theorem \ref{c5.t.num-umgsf}, let $Y$ be the finite minimal generating set of $H$.

\begin{theorem}\label{c5.t.intemp}
If $\gcd(Y) = 1$, then  $|H_{L\mbox{-}p}| = \infty$.
\end{theorem}

\begin{proof}
If $\gcd(Y) = 1$, by Theorem \ref{c5.t.numimonoidgcd1}, the submonoid $H$ is a numerical monoid so that $|\mathds{N}\setminus H| < \infty$. Thus, $H$ contains infinitely many prime numbers. Since $1 \notin L$, every prime number is $L$-primitive. Hence, $|H_{L\mbox{-}p}| = \infty$.
\end{proof}

\begin{theorem}\label{c5.t.les1-infy}
If $L$ is a finite set and $\gcd(Y) > 1$, then $|H_{L\mbox{-}p}| \le 1$ or $|H_{L\mbox{-}p}| = \infty$.
\end{theorem}

\begin{proof}
We first assume that $l$ $|\!\!\!\!\!\not$\;\; $d$, for all $l\in L$ and claim that $|H_{L\mbox{-}p}| = \infty$. Let $\gcd(Y) = d$. Since $d \neq 1$, by Proposition \ref{c5.t.numimonoidgcd1}, the submonoid $H$ is not a numerical monoid. We define the function \[f: H \longrightarrow \mathds{N}\;\;\mbox{ by}\;\; hf = \frac{h}{d}.\] Clearly, $f$ is a monomorphism and therefore the image of $f$, $\mbox{Im}(f)$, is isomorphic to $H$. By  Theorem \ref{c5.t.numimonoidgcd1}, the submonoid $\mbox{Im}(f)$ is a numerical monoid.

Clearly, $\mbox{Im}(f)$ has infinitely many prime numbers. Let $p \in \mbox{Im}(f)$ be a prime number such that $p > \max(L)$, then $pd \in H$. By Euclid's lemma, $l$ $|\!\!\!\!\!\not$\;\; $pd$, for all $l \in L$. Since $\mbox{Im}(f)$ has infinitely many such prime numbers, we have $|H_{L\mbox{-}p}| = \infty$.

Now, we assume that $l\;|\;d$, for some $l \in L$. Here, we determine $|H_{L\mbox{-}p}|$ with respect to $d \in L$ or not. If $d \notin L$, then clearly $|H_{L\mbox{-}p}| = 0$. If $d \in L$, we consider the cases $d \in H$ or not. If $d \notin H$, then clearly $|H_{L\mbox{-}p}| = 0$. In case $d \in H$, if there is an $l' (\ne d)$ which divides $d$, then $|H_{L\mbox{-}p}| = 0$; otherwise $|H_{L\mbox{-}p}| = 1$.
\end{proof}

\begin{remark}
If $L$ is an infinite subset of $\mathds{N}$, then $|H_{L\mbox{-}p}|$ need not satisfy the Theorem \ref{c5.t.les1-infy}. For instance, let $H$ be the submonoid of $\mathds{N}$ generated by the set $\{4, 6\}$ and $L =  \{4\} \cup \{\mathfrak{P}\setminus \{2,5\}\}$, where $\mathfrak{P}$ is the set of all prime numbers in $\mathds{N}$. We observe that $H_{L\mbox{-}p} = \{4,10\}$ and so $|H_{L\mbox{-}p}| = 2$. Similarly, if $L = \{4\} \cup \{\mathfrak{P}\setminus \{2,5,7\}\}$, then $H_{L\mbox{-}p} = \{4,10,14\}$ and so $|H_{L\mbox{-}p}| = 3$.
\end{remark}

In the following, we make certain remarks on the number of $L$-primitive words in the submonoids of a free monoid over an alphabet of size at least two. First observe that if the submonoid $H$ is $\{\varepsilon\}$, then $|H_{L\mbox{-}p}| = 0$. If $H \ne \{\varepsilon\}$, then by Corollary \ref{c5.c.prim-lprim}, we have the following remark.

\begin{remark}
If $H$ is a noncommutative submonoid of $A^*$, where $|A|\ge 2$, then $|H_{L\mbox{-}p}| = \infty$.
\end{remark}

\section{Conclusion}

Motivated by the work of Ito \emph{et al.} in \cite{ito88}, we have considered a study on the number of primitive words in the languages of semi-flower automata (SFA). SFA precisely accept finitely generated submonoids of free monoids \cite{shubh12}.  We extended the study to submonoids of free monoids and observed that the number is either at most one or infinite. Further, we have initiated a study on the number of $L$-primitive words in submonoids of free monoids. If $L$ is a finite, we have counted the number of $L$-primitive words in the submonoids of a free monoid  over a unary alphabet. When $L$ is infinite, the problem appears to be more complicated and a systematic study in this regard is necessary. In case the alphabet size is at least two, we could remark only on the number of $L$-primitive words in noncommutative submonoids. One can consider the problem in commutative submonoids.


\end{document}